\pdfoutput=1

\documentclass[twoside,leqno,twocolumn]{article}
\usepackage{ltexpprt}
\usepackage{framed}
\usepackage{balance}
\usepackage{epsfig}
\usepackage{color}
\usepackage{subfigure}
\usepackage{multirow,tabularx}
\usepackage{amssymb}
\usepackage{mathtools}
\usepackage{graphicx}
\usepackage{epstopdf}
\usepackage{cite}
\usepackage{multirow}
\usepackage{bm}
\usepackage{csquotes}
\usepackage{array}
\usepackage[yyyymmdd,hhmmss]{datetime}
\usepackage[subject={Todo}]{pdfcomment}
\usepackage[textsize=footnotesize,bordercolor=black!20]{todonotes}
\usepackage{xcolor,colortbl}
\usepackage{amssymb}
\usepackage{pifont}
\usepackage[algo2e,titlenumbered,ruled]{algorithm2e} 
\usepackage{lipsum,environ,amsmath}
\usepackage{slashbox,booktabs,amsmath}
\usepackage{hyphenat}
\usepackage{xspace}

\newcounter{Lcount}
\newcommand{\numsquishlist}{
   \begin{list}{\arabic{Lcount}. }
    { \usecounter{Lcount}
 \setlength{\itemsep}{-.1ex}      \setlength{\parsep}{0ex}
      \setlength{\topsep}{0ex}       \setlength{\partopsep}{0ex}
      \setlength{\leftmargin}{1em} \setlength{\labelwidth}{1em}
      \setlength{\labelsep}{0.1em} } }
\newcommand{\numsquishend}{\end{list}}

\newcommand{\squishlist}{
   \begin{list}{$\bullet$}
    { \setlength{\itemsep}{-.1ex}      \setlength{\parsep}{0ex}
      \setlength{\topsep}{0ex}       \setlength{\partopsep}{0ex}
      \setlength{\leftmargin}{.8em} \setlength{\labelwidth}{1em}
      \setlength{\labelsep}{0.5em} } }
\newcommand{\squishend}{\end{list}}

\definecolor{Gray}{gray}{0.85}
\definecolor{LightGray}{rgb}{0.9,0.9,0.9}
\definecolor{LightBlue}{rgb}{0.8,0.8,1}

\newcounter{problem}
\newenvironment{problem}[1][htb]
  {
   \begin{algorithm2e}[#1]%
  }{\end{algorithm2e}}

\makeatletter
\AtBeginEnvironment{problem}{\let\c@algocf\c@problem}
\makeatother

\newcommand{\flip}{{\sc Faction Initiator Inference Problem}\xspace}%

\begin{document}

\title{Framework for Inferring Leadership Dynamics of  Complex Movement  from Time Series}

\author{Chainarong Amornbunchornvej\thanks{Department of Computer Science, University of Illinois at Chicago. \{camorn2,tanyabw\}@uic.edu}
 \and Tanya Berger-Wolf\footnotemark[1]  \\
}
\date{}

\maketitle



\fancyfoot[R]{\footnotesize{\textbf{Copyright \textcopyright\ 2018 by SIAM\\
Unauthorized reproduction of this article is prohibited}}}





\begin{abstract} \small\baselineskip=9pt Leadership plays a key role in social animals, including humans, decision-making and coalescence in coordinated activities such as hunting, migration, sport, diplomatic negotiation etc. In these coordinated activities, leadership is a process that organizes interactions among members to make a group achieve collective goals. Understanding initiation of coordinated activities allows scientists to gain more insight into social species’ behaviors. However, by using only time series of activities data, inferring leadership as manifested by the initiation of coordinated activities faces many challenging issues. First, coordinated activities are dynamic and are changing over time. Second, several different coordinated activities might occur simultaneously among subgroups. Third, there is no fundamental concept to describe these activities computationally. In this paper, we formalize \flip and propose a leadership inference framework as a solution of this problem. The framework makes no assumption about the characteristics of a leader or the parameters of the coordination process.
The framework performs better than our non-trivial baseline in both simulated and biological datasets (schools of fish). Moreover, we demonstrate the application of our framework as a tool to study group merging and splitting dynamics on another biological dataset of trajectories of wild baboons.
In addition, our problem formalization and framework enable opportunities for scientists to analyze coordinated activities and generate scientific hypotheses about collective behaviors that can be tested statistically and in the field.

 \hspace{0mm}{\bf Keywords}: leadership, coordination, time series, influence, collective behavior
\end{abstract}

\section{Introduction}
Leadership is a process of individuals (leaders) who influence a group to achieve collective goals~\cite{hogg2001social,glowacki2015leadership}.  Leadership plays a key role in solving collective-action problems (e.g. social conflicts, migration, hunting, territorial defense) across social species~\cite{glowacki2015leadership},  organizing collective movement~\cite{couzin2005effective}, as well as collaboration in group's decision making~\cite{Dyer:2009aa,glowacki2015leadership}. In the context of coordination, which is defined as an emergence of collective actions to achieve the collective goals~\cite{malone1994interdisciplinary}, leadership mainly contributes by fostering collective behaviors in social species ranging from humans ~\cite{Dyer:2009aa,hogg2001social,glowacki2015leadership} to fish~\cite{krause2000leadership}. 

In nature, leadership can be viewed as a process of initiation of coordinated activity. For example, leadership is a process by which leaders initiate the group's coordinated movement toward a destination~\cite{krause2000leadership,Smith2015187,stueckle2008follow}. In this process, leaders guide their group's members to follow in the right direction. Understanding how leaders emerge and influence collective behaviors enables scientists to gain insight into synchronization and coordination processes in nature. In this paper, we use the words {\bf `leader'} and {\bf `initiator'} interchangeably.

While many studies on leadership in coordinated activity exist in behavioral research, there are a few computational approaches addressing the leadership of coordination. In social network analysis, Influence Maximization (IM)~\cite{kempe2003maximizing, goyal2010learning,He2016maximizing} is one of the classic problems that focuses on inferring a subset of individuals that maximizes information spreading. However, IM focuses solely on finding potential initiators who initiate the coordination of {\em information spreading} and, moreover, does not address the question of {\em when} coordination happens. The method for inferring leaders from online communities actions~\cite{goyal2008discovering} can be used to identify the {\em group} being coordinated but it, still, does not provide the information on {\em when} coordinated activities happen. In movement coordination, ~\cite{andersson2008reporting, kjargaard2013time, PhamICDE2016, carmi2013inferring, jacoby2016inferring} propose methods specific to movement activity for finding leaders {\em during} group's  movement intervals but none of them can be used to identify the time of the process of coordination. There also exist many works regarding collective behavior and implicit leaders~\cite{Wu2014crowdmodel,couzin2005effective,Yu:2010:CDM:1838186.1838192}. In this model, leaders can influence their group implicitly and leaders' identity might be unknown to the group.  Still, none of the works in this category can be used to infer the time of the periods of coordination. 

Since leadership is a collective process~\cite{hogg2001social}, considering only dyadic interactions is not enough to infer a leadership instance. Therefore, the works in \cite{flica,kjargaard2013time,jacoby2016inferring} proposed leadership frameworks that are based on a network representation of time series.

In the context of coordination leadership, the method of leadership inference in~\cite{flica} provides a solution for  identifying coordination events,  the initiators of these events, as well as  proposes an approach for the classification of the types of leadership models acting on a group. However, the framework in~\cite{flica} cannot be used to infer {\em multiple} coordinated activities which can occur simultaneously because the notion of multiple factions is not employed by the framework.  We aim to close these gaps in the study of coordination leadership.

\subsection{Our Contributions.}
First, we introduce the novel computational problem of leadership identification in multiple coordinated activities, namely \flip. We formalize the problem and analyze its theoretical properties and implications. Second, we propose a framework for \flip by combining several existing methods in a principled and novel manner. Our framework is capable of:
\squishlist
\item {\bf Detecting intervals of multiple coordination:} inferring intervals when different coordinated activity in one or more groups may occur simultaneously; 
\item {\bf Identifying leaders:} identifying the initiators of these coordinated activities, the individual who initiates each coordination and the group that follows; 
\item {\bf Discovering the events of merging and splitting of coordination:} identifying the time when a coordinated group is separated into smaller sub-groups or merged with another coordinated group.
\squishend

\begin{framed}  
\noindent {\flip:} { To reach collective goals, group's members must coordinate with each other. Multiple factions within a big group may exist solving their sub-tasks in helping the entire group achieve the collective goals. {\bf Given time series of individual activities, our goal is to identify periods of coordination and the subsequent coordinated activity, find factions of coordination if more than one exist, as well as identify leaders of each faction}}
\end{framed}

We demonstrate the ability of the framework to infer leadership in multiple coordinated groups on both simulated and biological datasets. Since we propose the new problem and framework and no other approaches exist, we compare our framework against a non-trivial baseline, which is the modification of the closest existing approach in leadership inference.  Our approach is flexibly generalizable to any multiple coordinated activities from any time series data. 

\subsection{Influence Maximization vs. \flip}
Influence Maximization can be viewed a special case of the \flip, namely a single event of coordinating the state of information in a social network, using a specific coordination (spreading) mechanism. 
  \squishlist 
  \item {\bf Coordination Mechanism:} Majority of Influence Maximization work uses Independent Cascade and Linear Threshold models as main coordination mechanisms. Yet, there are other models, such as Hierarchy, Dictatorship or other non-network based models that can be represented as coordination mechanisms. The new problem we formalize in this paper, \flip, generalizes to all mechanisms for coordinating group activities and we demonstrate so in this paper by using datasets generated by several models of coordination mechanisms.  
  \item {\bf Coordination Event:} Influence Maximization focuses mainly on an information spreading event happening in a social network. The information state for each node are the time series being coordinated. However, this is one particular type of a coordination event and other, more general and non-network, coordination activities are possible. For example, a coordinated movement activity of animals is a coordination event that has animals coordinating their trajectories, not necessarily through a wave-like spread of information in a network, to reach a group destination.  Our proposed framework can handle all types of coordination events, including but not limited to network information spreading.  
  \item {\bf The dynamics of coordination:} In influence Maximization, majority of papers focus on inferring a single global set of initiators that can maximize influence in a given network. However, in a single dataset, there can be many coordination events and each event can have different initiators. Moreover, coordination events with different initiators might happen simultaneously. The framework we propose here aims to address the dynamics of coordination from data and is capable of inferring when each coordination event happens and who are the initiators.     
\squishend
\newcommand{\argmax}{\mathop{\mathrm{argmax}}\limits} 

\newtheorem{definition}{Definition}

\section{Problem statement and analysis}
\subsection{Coordination without noise.}
Given a collection of time series, our goal is to find multiple coordination intervals as well as their initiators. We do not assume that the coordination intervals that belong to different coordinated sets of time series are disjoint and allow overlap. We formalize various concepts of coordination and following similar to~\cite{flica}.

\begin{definition}[Following relation]
\label{followRDef}
Let $U = (\vec{u}_1,\dots,\vec{u}_t,\dots)$ and $W = (\vec{w}_1,\dots,\vec{w}_t,\dots)$ be arbitrary-length time series. If $\forall t \in \mathbb{N}$, there exists a time delay $\Delta t \in \mathbb{Z^+} \cup\{0\}$, such that $\vec{w}_t=\vec{u}_{t+\Delta t}$, then $U$ follows $W$, denoted as $W \preceq U$ for any $\Delta t$ and $W \prec U$ if $\Delta t>0$. 
\end{definition}


\begin{definition}[Coordination]	
\label{CoorDef}
Given a set of $m$-dimensional time series $\mathcal{U} = \{U_1, \dots, U_n\}$. The set $\mathcal{U}$ is {\em coordinated} at time $t$ if for every ${n \choose 2}$ pairs $U_i, U_j \in \mathcal{U}$, there exists either $U_i \prec U_j$ or $U_j \prec U_i$. 
The {\em coordination interval} is the maximal contiguous time interval $[t_1, t_2]$ such that $\mathcal{U}$ is coordinated for every $t \in [t_1, t_2]$.
\end{definition}

\begin{definition}[Initiator] 
\label{InitiatorDef}
Let $\mathcal{U} = \{U_1, \dots, U_n\}$ be a set of coordinated $m$-dimensional time series within some coordination interval $[t_1, t_2]$. Then the time series $L \in \mathcal{U}$ is the {\em initiator}  time series for the coordination interval if for each time series $U \in \mathcal{U} \setminus \{L\}$, $L \prec U$.  
\end{definition}

\begin{definition}[Following network]
\label{followNetDef}
Let $\mathcal{U} = \{U_1, \dots, U_n\}$ be a set of time series, a directed graph $G = (V,E)$ is defined as a following network, where $V$ is a set of nodes that has a one-to-one correspondence to the time series set $\mathcal{U}$ and $E$ is a set of edges, such that $e_{i,j} \in E$ if $U_j \prec U_i$.
\end{definition}


We now extend these concepts to the case of multiple coordinated subgroups.

\begin{definition}[Faction] 
Given a set of time series $\mathcal{U}$, a subset $F \subseteq \mathcal{U}$ at time $t$ is maximally coordinated, if  $F$ is coordinated and there is no other coordinated set $F' \subseteq \mathcal{U}$ where $F \subset F'$. We call such maximally coordinated $F$ a faction at time $t$. 
\label{FactionDef}
\end{definition}

\begin{definition}[Faction interval] 
The {\em coordination interval} of a faction $F$ or a faction interval is the maximal consecutive time interval $[t_1, t_2]$ such that $F$ is coordinated for every $t \in [t_1, t_2]$.
\label{FactionInvDef}
\end{definition}

Faction is a structurally maximal subset and its interval is a temporally maximal subset.

\begin{lemma}
\label{FracMemLem}
A time series $W$ is a member of a faction $F$ if and only if it has an edge to $F$'s initiator $L$. 
\end{lemma}
\begin{proof}
Let a time series $W \in F$. Since $\forall U \in F\setminus \{L\}$, $L \prec W$. By definition, there is an edge from $W$ to $L$.\\

Let $L\prec W$.  If $W$ is not in $F$, then we can add $W$ to $F$, which will remain a coordinated set but will now violate the maximality of $F$. Thus, $W\in F$.
\end{proof}

According to Lemma~\ref{FracMemLem}, a faction $F$ is a set of nodes within $G$ such that all nodes within $F$ have a directed edge to $L$. Note that $L$ always has the out-degree of zero and in-degree of $|F|-1$ within a coordination interval.

\begin{problem}
    \SetKwInOut{Input}{Input}
    \SetKwInOut{Output}{Output}
    \Input{Set  $\mathcal{U} = \{U_1,\dots, U_n\}$ of $m$-dimensional time series}
    \Output{ A set of factions $\mathcal{F}=\{F_1,\dots,F_k\}$, a set of coordinated intervals $\mathcal{T}=\{[t^1_1,t^1_2],...,[t^k_1,t^k_2]\}$, and the set of initiator time series $\mathcal{L}=\{L_1,...L_k\}$ where $L_i$ initiated the coordination interval $[t^i_1,t^i_2]$ of the faction $F_i$}
    \caption{{\flip}}
    \label{FLIPProb}
\end{problem}

We are now ready to formally state the \flip at Problem~\ref{FLIPProb}.

\subsection{Coordination with noise.}
In the previous section, we stated the definitions and properties of the problem of identifying multiple faction initiators  in the ideal setting. In this section, we provide the relaxation and the analysis of the problem in the presence of noise. 

\begin{definition}[$\sigma$-following relation] 
Let $\mathcal{U}$ be a set of time series  and $\mathrm{sim}: \mathcal{U}\times \mathcal{U} \to [0,1]$ be some similarity measure between two time series. For any pair of time series $U_i,U_j \in \mathcal{U}$, let $\Delta t_{max} =\mathrm{min}\argmax_{\Delta t \in \mathbb{Z}} \mathrm{sim}(U_{i,1},U_{j,1+\Delta t})$ where $U_{i,t}$ represents a time series $U_i$ starting at time $t$, and let $\mathrm{sim}_{max}(U_i,U_j)=\mathrm{sim}(U_{i,1},U_{j,1+\Delta t_{max}})$. Then, for a threshold $\sigma \in (0,1]$, if $ \mathrm{sim}_{max}(U_i,U_j) \geq \sigma$, then we have: 
\squishlist 
  \item if $\Delta t_{max} >0$ , then ${U_i \prec_{\sigma} U_j}$,
  \item if $\Delta t_{max} <0$ , then ${U_j \prec_{\sigma} U_i}$,
	\item if either $\Delta t_{max} = 0$ or ${U_i \prec_{\sigma} U_j}$ and ${U_j \prec_{\sigma} U_i}$, then ${U_i \equiv_{\sigma} U_j}$ ($U_i$ is $\sigma$-following equivalent to $U_j$).
\squishend
\label{SigmaFollwDef}
\end{definition}

Note that if two time series $U$ and $W$ such that ${U \prec_{\sigma} W}$ and ${W \prec_{\sigma} U}$,  there exists more than one position in time $\Delta t_{max}$ that make both time series maximize their similarity. 

\begin{definition}[$\sigma$-coordination]
Let $\mathcal{U}$ be a set of time series, then $\mathcal{U}$ is $\sigma$-coordinated if for every ${|\mathcal{U}| \choose 2}$ pairs $U_i, U_j \in \mathcal{U}$, either ${U_i \prec_{\sigma} U_j}$ or ${U_j \prec_{\sigma} U_i}$ exists. 
\end{definition}

\begin{definition}[$\sigma$-faction]
Let $\mathcal{U}$ be a set of time series. A $\sigma$-faction $F \subseteq \mathcal{U}$ is a maximal set such that $F$ is $\sigma$-coordinated, and there is no other $\sigma$-coordinated set $F' \subseteq \mathcal{U}$ where $F \subset F'$.
\label{def:sigmafactions}
\end{definition}  


\begin{definition}[Relaxed faction interval]
\label{def:rexfacinterval}
Let $\mathcal{U}$ be a set of time series, the time interval $[t_1,t_2]$ is a faction interval of initiator $L$ if for all $t \in [t_1,t_2]$, there exists a faction $F_t$ such that $F_t$ has $L$ as its initiator and $|F_t| >1$.
\end{definition}

\subsection{Coordination measure}
Given a set of time series $\mathcal{U}$, a set of clusters $\mathcal{C}=\{H_1,\dots,H_n\}$ such that $\bigcup_k H_k=\mathcal{U}$,  we define a cluster membership indicator $\delta_{i,j} =1$ if time series $U_i$ and $U_j$ belong to the similar cluster, otherwise it is zero.  {\bf The average coordination measure} $\Psi$ of a set of clusters $\mathcal{C}$ is defined as follows:

\begin{equation}
\label{eq:coor-measure}
\Psi(\mathcal{C}) = \frac{\displaystyle\sum_{ U_i,U_j \in \mathcal{U},U_i \neq U_j }\mathrm{sim}_{max}(U_i,U_j)\delta_{i,j}}{\displaystyle\sum_{ U_i,U_j \in \mathcal{U},U_i \neq U_j }\delta_{i,j}}.
\end{equation}

Note that $\Psi \in [0,1]$. If $\Psi$ is close to 1, then all time series within the same cluster are highly similar, with some time delay. This implies a high degree of coordination within each cluster in this case. On the contrary, $\Psi \approx 0$ implies no coordination, on average.

\begin{theorem}
Given a set of time series $\mathcal{U}$ containing a set of $\sigma$-faction $\mathcal{F}=\{F_1,\dots,F_n\}$ where $\bigcup_{F_i \in \mathcal{F}}F_i=\mathcal{U}$, then, for all possible sets of clusters, $\mathcal{F}$ maximizes the average coordination measure $\Psi$. 
\end{theorem} 

Proof of Theorem~\ref{eq:coor-measure} is in the supplementary material.

\section{Methods}
\label{sec:method}

\begin{figure}[ht!]
\centering
\includegraphics[width=1\columnwidth]{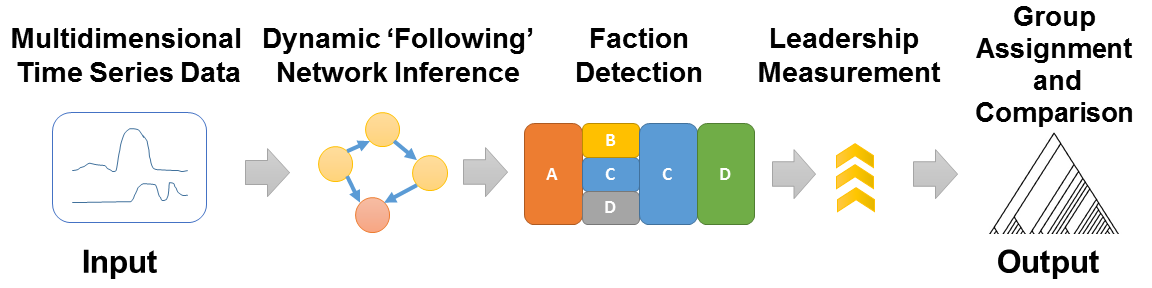}
\caption{A high-level overview of mFLICA framework}
\label{fig:processingSteps}
\end{figure}

We propose the following framework to solve \flip. The framework is designed to infer a set of factions, faction intervals, and their initiators from time series. Figure~\ref{fig:processingSteps} depicts the overview of our framework.


\subsection{Following network inference.}
\label{subsec:following-network-inference}

Given a set of time series $\mathcal{Q}$ and a similarity threshold $\sigma$, for each pair of time series $U,W \in \mathcal{Q}$, our goal is to measure whether either $U$ follows $W$ or no following relation between them exists. The time series similarity measure we need should satisfy the following properties. First, it should recognize common patterns between two time series if they exist. These common patterns can be noisy, distorted, time-delayed, and discontinuous.  Second, it should infer time delay between these common patterns.  

We deploy Dynamic Time Warping (DTW) \cite{sakoe1978dynamic} as the similarity measure of following relation since DTW's warping path can distinguish whether two time series share  noisy common patterns and can approximately infer the time delay of common patterns between time series. Besides, according to the work in \cite{kjargaard2013time}, DTW performance is superior to that of other methods in detecting following among time series. 

For any pair of time series $U,W \in \mathcal{Q}$, we use the equation from~\cite{flica}, to approximate a following relation as below:


\begin{equation}
	\mathrm{s}(P_{U,W})=\frac{\sum_{(i,j) \in P_{U,W}}\mathrm{sign}(j-i)}{|P_{U,W}|},
	\label{eq:traCorr}
\end{equation}

where $P_{U,W}$ is the optimal path of DTW. If $(i,j) \in P_{U,W}$, then $U$ at time $i$ is the most similar to $W$ at time $j$. When $-\sigma<\mathrm{s}(P_{U,W})<\sigma$, neither $U$ nor $W$ follows each other.  We have ${U \prec_{\sigma} W}$ if  $\mathrm{s}(P_{U,W}) \geq\sigma$. In contrast, $\mathrm{s}(P_{U,W}) \leq-\sigma$ implies ${W \prec_{\sigma} U}$. The function is bound by $\mathrm{s}(P_{U,W}) \in [-1,1]$ and we set $\sigma = 0.5$ for our framework as default.

Then, a following network $G = (V_\mathcal{Q},E)$ is constructed from $\mathcal{Q}$ where $v_k \in V_Q$ represents a vertex of time series $Q_k \in \mathcal{Q}$ and $E_{k,l}=|\mathrm{s}(P_{Q_k,Q_l})|$ if ${Q_l \prec_{\sigma} Q_k}$. The pseudo code of following network inference is in the supplementary material.

\subsection{Dynamic Following network inference.}
\label{subsec:Dy-following-network-inference}
As mentioned before, a set of time series $\mathcal{U}$ might consist of multiple overlapping coordination intervals from many factions. Using only summary statistics or an aggregate following network cannot discover these dynamics. Hence, we need to consider each local interval and build a following network to represent the interval. Therefore, we deploy a dynamic following network scheme in our framework, which is a common technique to deal with dynamics of  data\cite{holme2014temporal}.

The next question is ``how long should each local interval be?'' For now, we assume that we have a priori knowledge of the time window $\omega$ to capture local intervals. Later we show that  we can infer $\omega$ from the dataset itself in Section~\ref{subsec:TWinference}.  

We have a set of $t^*$-length time series $\mathcal{U}$ as the input.   We sample all time series within $\mathcal{U}$ by sliding window intervals and create following networks of these intervals. Let $\omega \in \mathbb{N}$ be a time window and $\delta = 0.1\omega$ (time shift threshold), the $i$-th sliding window interval, be defined by: $w(i) = [(i-1)\times\delta,(i-1)\times\delta+\omega]$. For each interval $w(i)$, we have a set of time series $\mathcal{Q}$.  For each time series $U_k \in \mathcal{U}$, there is $Q_{k} \in \mathcal{Q}$ such that $Q_k$ is a subset of $U_k$ during $w(i)$ time interval. We build a following network for each $w(i)$, then we combine these networks to be a single dynamic network. The pseudo code of the dynamic network creation is in the supplementary material.

\subsection{Factions detection and coordination intervals.}
\label{subsec:factions-detection}

For each following network $G=(V,E)$, factions are network components such that all member nodes directly connect to their initiator (Lemma~\ref{FracMemLem}). We infer factions based on Definition~\ref{def:sigmafactions} and the coordination intervals of factions are discovered based on Definition~\ref{def:rexfacinterval}. 

According to Lemma~\ref{FracMemLem}, initiator nodes have outgoing-degree zero, and all nodes within the similar faction directly connect to their initiator. However, due to the introduction of the time window $\omega$, some nodes might not have direct edges to the initiators. Therefore, we relax the constraint of faction membership to make all nodes which have any directed path to an initiator to be members of the initiator's faction. 

Since a faction is a directed connected component where all nodes are reachable from the initiator by inverse paths, we use Breadth-First Search (BFS) to identify all reachable nodes from each initiator node in the following network in order to find members of each faction.  
The pseudo code of this step is in the supplementary material.

A useful statistic about factions (used later) is the {\em faction size ratio}. Let $G_l=(F_l,E_l)$ be an induced subgraph of $G$ defined by faction $F_l$, then the faction size ratio of $F_l$ is defined as follows:

\begin{equation}
\text{fs}(F_l)= \frac{|E_l|}{ {|V| \choose 2}}.
\label{eq:FZratio}
\end{equation}

\subsection{Time window inference.}
\label{subsec:TWinference}

In reality, some following relations might not be cause by explicit initiators since they either happen by chance or are due to other factors which are not related to the influence of leaders. For instance, if a follower is unable to observe a leader's pattern, then the follower cannot be influenced by the leader. Different types of time series have different limitation of `observation memory', which is the limitation of time delay $\Delta t$ such that a follower can truly observe and imitate its leader's actions or can get commands from a leader.  

Hence, to represent the concept of observation memory limitation, we set the time window $\omega$ to limit the length of the time delay $\Delta t$ that  can measure following relations. Moreover, $\omega$ helps us prevent the comparison of time series between different coordination events. 

Nevertheless, if we set $\omega$ too small, we miss inferring some following relations that have $\Delta t>\omega$. On the contrary, long-length $\omega$ causes false positive matching between repeated patterns of different coordination intervals. Therefore, a proper $\omega^*$ should be able to infer a higher number of true following relations than any $\omega$. Even if some random following relations might appear when we choose $\omega$ instead of $\omega^*$, this is not an issue.  Since these random following relations appear by chance and with lower probability, they have a relatively small effect on the number of following relations.  

In our framework, without the knowledge of $\omega$, we use $\omega$ that maximizes the average coordination measure $\Psi$ (Eq.~\ref{eq:coor-measure}). 
Given a dynamic following network based on the time window $\omega$, for each time step $t$, we calculate $\Psi_t$ by designating each faction to be a cluster and creating the last cluster for all time series, which are not in any faction. Then, $\hat{\Psi}_{\omega}$ is computed from the median of $\{\Psi_1,...,\Psi_{t^*}\}$. $\hat{\Psi}_{\omega}$ is used to be a representative coordination measure value of $\omega$. Hence, the optimal $\omega^*$ is computed as follows:

\begin{equation}
\omega^*= \argmax_{\omega} (\hat{\Psi}_{\omega}).
\end{equation}

\subsection{Leadership comparison.}
\label{subsec:leadership-comp}
There are several methods that are widely used for ranking important nodes within the graph. One of the well-known methods that consider the higher-order relation within a graph is PageRank~\cite{page1999pagerank}. 
In our approach, we deploy PageRank on the following network to rank individuals within each faction and report the rank ordered lists for each time step. Even though PageRank scores are computed from the entire network, we compare individuals' ranking score only within the same faction and create a rank order list for each faction. For each node $i$ within a following network $G$, the PageRank score is defined below:

\begin{equation}
\pi_i=d\sum_{k \in \mathcal{N}^{in}_i} \frac{E_{k,i}\pi_k}{|\mathcal{N}^{out}_k|}+(1-d),
\label{eq:PReq}
\end{equation}

where $\pi_i \in [0,1]$ is a rank value of node $i$, $d$ is a damping factor, which is typically set at 0.9, $\mathcal{N}^{in}_i$ is a set of $i$'s followers, $\mathcal{N}^{out}_k$ is a set of individuals $k$ follow, and $E_{k,i} \in [0,1]$ is an element of adjacency matrix of a following network where $k$ follows $i$ if $E_{k,i}\geq \sigma$. 

\section{Evaluation Datasets}
\label{sec:experiment}

\subsection{Leadership models.}
\label{sec:LeadershipModels}
The evaluation of the framework is conducted based on four models of coordination mechanisms.
\subsubsection{Dictatorship Model.}
The Dictatorship Model (`DM')~\cite{flica} is considered to be the simplest model in the leadership realm. Initially, no movement happens until the leader starts moving to a target, then individuals follow their leader with some time delay until the entire group is coordinated in both direction and velocity. Then, the group gradually stops at the target and starts moving again to the next target.   

\subsubsection{Hierarchical Model.}
The Hierarchical Model (`HM')~\cite{flica} is another variation of DM with the hierarchical condition. The hierarchical condition assigns a rank to each individual within a group. A leader has the highest rank. The low-rank individuals follow high-rank individuals with some time delay. In our evaluation model, we assign a linear order hierarchical condition such that $\text{ID}(1)$ is a leader and $\text{ID}(n)$ is followed by $\text{ID}(n+1)$. The group moves linearly along the line with some noise, following its leader.
 
\subsubsection{Independent Cascade Model.}
The Independent Cascade Model (`IC')~\cite{kempe2003maximizing} is one of the influence propagation models in Social network analysis. Initially, everyone has a probability to be activated $\rho$. Active individuals move toward their leader. For each time step, each active individual simultaneously and independently attempts to activate $k$-nearest inactive neighbors around it  with the probability of success $\rho$. If success, the inactive individual becomes active at the next time step. Active individuals cannot attempt to activate the same individuals again. Only the leader follows its target and everyone else follows the leader. We explore the parameter space on combinations of: $k \in \{3,5,10\}$ and $\rho \in \{0.25,0.50,0.75\}$.

\subsubsection{Crowd Model.}
In the Crowd Model (CM)  \cite{Wu2014crowdmodel}, there are two types of individuals: informed and uninformed individuals. For each time step, informed individuals move toward the target independently while uninformed individuals keep staying close to both group's position and direction centroids. Therefore, the group direction is implicitly influenced by informed individuals. For each coordination, all informed individuals follow a single target direction vector, while the rest of the group keeps staying with the majority. 

\subsection{Synthetic trajectory simulation.}

We generate time series datasets based on the models described above. For each dataset, it consists of 30 individuals' time series of $X,Y$ coordinates. Each time series has a length of 4,000 time steps. A coordination event consists of multiple faction intervals, described below. We have five coordination events for each dataset. For each model above, the coordination event can be divided into two types.

\begin{figure}[ht!]
\centering
\includegraphics[width=1\columnwidth]{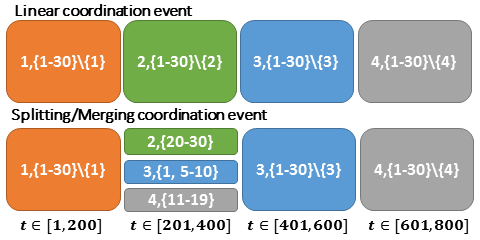}
\caption{Linear (above)  and  Splitting/Merging (below) coordination event. Each block represents a faction such that the first element is the leader ID and the second element is the member IDs set. The time interval each faction appears is at the last line.}
\label{fig:CoEvTypes}
\end{figure}

\subsubsection{Linear coordination event.}
\label{sec:LinearCo-event}
There are four factions for each coordination event. The first faction has $\text{ID}(1)$ as a leader and others are followers. This faction lasts for 200 time-steps. The next faction is lead by $\text{ID}(2)$ and its coordination interval is $[201,400]$. The third faction appears within $[401,600]$ interval and it has $\text{ID}(3)$ as a leader. In the last faction, $\text{ID}(4)$ leads the group to stop moving and the group completely stops moving around time step $t=700$. Everyone stops moving within $[700,800]$, then the group proceeds to the next coordination event again. 

\subsubsection{Splitting/Merging coordination event.}
\label{sec:MSCo-event}
In this type of coordination event, splitting and merging of factions happens. Within the $[1,200]$ interval, $\text{ID}(1)$ leads a single faction with its direction vector. Then, at $t=201$, the group is split into three factions and they appear within $[201,400]$ interval. The first faction is lead by $\text{ID}(2)$ and about a third of the previous faction members are followers (Fig~\ref{fig:CoEvTypes} below). The $\text{ID}(2)$ has its own direction vector.  $\text{ID}(3)$ leads the second faction with another one third members from the previous faction. $\text{ID}(3)$ has a different direction from $\text{ID}(2)$. Lastly, $\text{ID}(4)$ leads the rest of the individuals. $\text{ID}(4)$ also has its own direction, which is different from $\text{ID}(2)$'s and $\text{ID}(3)$'s. 

At $t=401$, the factions lead by $\text{ID}(2)$ and $\text{ID}(4)$ are merged into the faction of $\text{ID}(3)$; $\text{ID}(2)$ and $\text{ID}(4)$ follow the $\text{ID}(3)$'s direction. At the $[401,600]$ interval, $\text{ID}(3)$ leads all the individuals. Finally, $\text{ID}(4)$ leads the faction to stop moving between $t=601$ and $t=700$. The group completely stops at the $[701,800]$ interval. Note that leaders in each faction are informed individuals in the Crowd Model. Instead of having only one leader for each faction, we have three informed individuals in the Crowd Model.\\

For each leadership model and its coordination event type, we generated 100 datasets. In total, each model has 200 datasets except IC, for which we explore all nine possible combinations of parameters. In total, we have 1,800 datasets for the IC model. 

\subsection{Biological datasets}
\subsubsection{Baboon trajectories.}
\label{subsubsec:baboon}

The dataset comes from the set of GPS collars on a troop of wild olive baboons (\emph{Papio anubis}) at Mpala Research Centre, Kenya \cite{crofoot2015data,strandburg2015shared}. The data consists of time series of latitude-longitude location pairs for each individual every second. $16$ individuals whose collars remained functional throughout the time are analyzed for a case study of a merging coordination event.

\subsubsection{Fish schools trajectories.}
\label{subsubsec:fish}

The fish dataset is a set of time series of fish positions from a video record of a school of golden shiners (\emph{Notemigonus crysoleucas}).  The record is used to study information propagation over the visual fields of fish \cite{strandburg2013visual}. Each trial contains  $70$ fish, with $10$ fish who trained to lead the group to the feeding sites. The dataset has $24$ separate ground truthed leadership events. The task is to correctly identify trained fish.

\section{Evaluation criteria}
\label{sec:Eval}
For each simulation dataset, we have the ground truth of an individual's membership in a faction and the identity of the faction's leader. We compared the inference result from each method against the known ground truth to evaluate the method's performance.
\subsection{Individual assignment.}
For all models, for each time step, the accuracy of the individual assignment is the number of  inferred individuals' factions that agree with the ground truth, divided by the total number of individuals. Note that, in the Crowd Model, each faction $F$ has a set of informed individuals and individuals belong to $F$ if they follow any informed individual in $F$.   
\subsection{Leadership prediction.}
For all models except the Crowd Model, the true positive $\text{TP}$ is the number of inferred leaders who are indeed the ground truth leaders. The false positive $\text{FP}$ is the number of inferred leaders who are not the actual leaders. The false negative $\text{FN}$ is the number of actual leaders who are inferred to be non-leaders. In the Crowd Model, $\text{TP}$ is the number of inferred leaders who are informed individuals from the right faction. $\text{FP}$ is the number of leaders who are uninformed individuals. $\text{FN}$ is the number of ground truth factions such that all informed members are non-leaders.  We calculated F1-Score to estimate the performance of the leadership prediction for each framework. 
\section{Results}
\subsection{Leadership Identification.}
\begin{table}
\caption{Factions and Leaders identification on simulation models}
\label{SimPredictTB}
\begin{tabular}{ |c|c|c|c|c|   }
 \hline
\multicolumn{1}{|c|}{} & \multicolumn{2}{c|}{Leadership F1-score} & \multicolumn{2}{c|}{Assignment Acc.} \\
\cline{2-5}
 Dataset& mFLICA & FLOCK  & mFLICA & FLOCK\\
 \hline
DM-L	&\bf{0.94}	&0.92	&\bf{0.89}	&0.86 \\ \hline
DM-MS	&\bf{0.94}	&0.91	&\bf{0.86}	&0.84 \\ \hline
HM-L	&\bf{0.94}	&0.91	&\bf{0.94}	&0.86 \\ \hline
HM-MS	&\bf{0.95}	&0.90	&\bf{0.86}	&0.81 \\ \hline
IC-L	&\bf{0.91}	&0.86	&\bf{0.86}	&0.80 \\ \hline
IC-MS	&\bf{0.89}	&0.85	&\bf{0.79}	&\bf{0.79} \\ \hline
CM-L	&\bf{0.82}	&0.64	&\bf{0.83}	&0.64 \\ \hline
CM-MS	&\bf{0.75}	&0.67	&\bf{0.64}	&0.55 \\ \hline
\end{tabular}
\end{table}

For each simulation model in Section~\ref{sec:LeadershipModels}, we evaluated results from all datasets using the criteria in Section~\ref{sec:Eval}. We set $\omega$ time window by the method from Section~\ref{subsec:TWinference} and set time shift $\delta=0.1\omega$. The results of faction assignments and leaders identification are in Table~\ref{SimPredictTB}. Each row with the label `-L' is a model with Linear coordination event type (Section~\ref{sec:LinearCo-event}) and `-MS' represents a model with Splitting/Merging coordination event type (Section~\ref{sec:MSCo-event}). The 2nd and 3rd columns represent the results of leadership prediction F1-Scores of mFLICA (our proposed framework) and the modified FLOCK framework~\cite{will2016flock,andersson2008reporting}, and the values in these columns are calculated from the median of all datasets from a given leadership model. The 4nd and 5rd columns represent individual assignment accuracy results. We took the median of all given-model datasets to represent each model accuracy. Unsurprisingly, mFLICA beat FLOCK in all models. The result implies that the simple framework like FLOCK has a limitation when it needs to deal with complicated noisy leadership models. \\

\begin{table}
\caption{Rank orders median accuracy within factions}
\label{HircRankTB}
\begin{center}
\begin{tabular}{ |c|c|c|   }
 \hline
\multicolumn{1}{|c|}{} & \multicolumn{2}{c|}{Top3 Rank Order Accuracy}  \\
\cline{2-3}
\multicolumn{1}{|c|}{Dataset} & \multicolumn{1}{c|}{mFLICA} & \multicolumn{1}{c|}{FLOCK} \\
\hline
HM-L	 &0.75	&0.78 \\  \hline
HM-MS	 &0.72	&0.76\\  \hline
\end{tabular}
\end{center}
\end{table}

 In hierarchical models, we reported the result of top 3 rank order inference accuracy within each faction in Table~\ref{HircRankTB}. The table rows represent leadership model datasets. The columns are accuracy, which determined by the percentage of top-$3$ individuals from the ground truth appear in the list of top-$3$ inferred list. Even though mFLICA has a competitive results, the FLOCK framework performs better, which makes sense since the hierarchical model has a linear hierarchy structure and the leader is always in the front of the group's direction, which matches the fundamental assumption of FLOCK.

\subsection{Case study: trained leaders in fish schools. }
\begin{table}
\caption{A school of fish inference median accuracy over 24 trails}
\label{FishTB}
\begin{center}
\begin{tabular}{ |c|c|c|  }
 \hline
   &  Trained fish  &Trained fish   \\
 Method &   factions & leaders  \\
 \hline
mFLICA	&0.90 &0.88 \\  \hline
 FLOCK	&0.37 &0.27	\\  \hline
\end{tabular}
\end{center}
\end{table}

We considered any fish within the faction of a trained fish to be following the trained fish. Among 24 trails of fish movement, the medians of inference accuracy of a fish following the trained fish are in column 2 in Table~\ref{FishTB}.	We also measured the accuracy of inferred initiators being the trained fish in each trial (column 3 in Table~\ref{FishTB}). According to the results in Table~\ref{FishTB}, mFLICA performs significantly better than FLOCK in both aspects. This is because the fish datasets are tremendously noisy, and the DTW in mFLICA is more robust to the noise than the simple FLOCK model~\cite{kjargaard2013time}. 

\subsection{Case study: detecting the group merging event of baboons. }
\begin{figure}[h!]
\centering     
\subfigure{\label{fig:baboon-exbottom}\includegraphics[width=1\columnwidth]{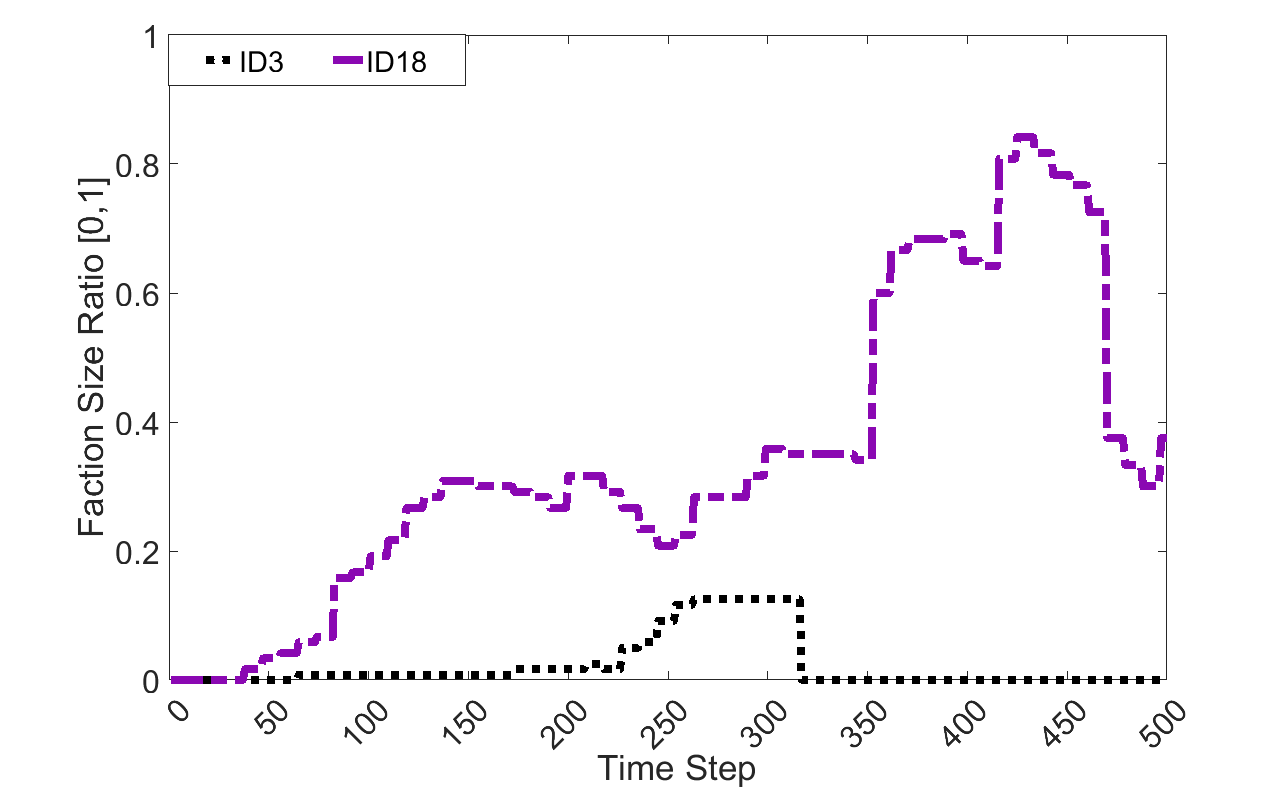}}
\subfigure[$t=300$]{\label{fig:baboon-exa}\includegraphics[width=0.32\columnwidth]{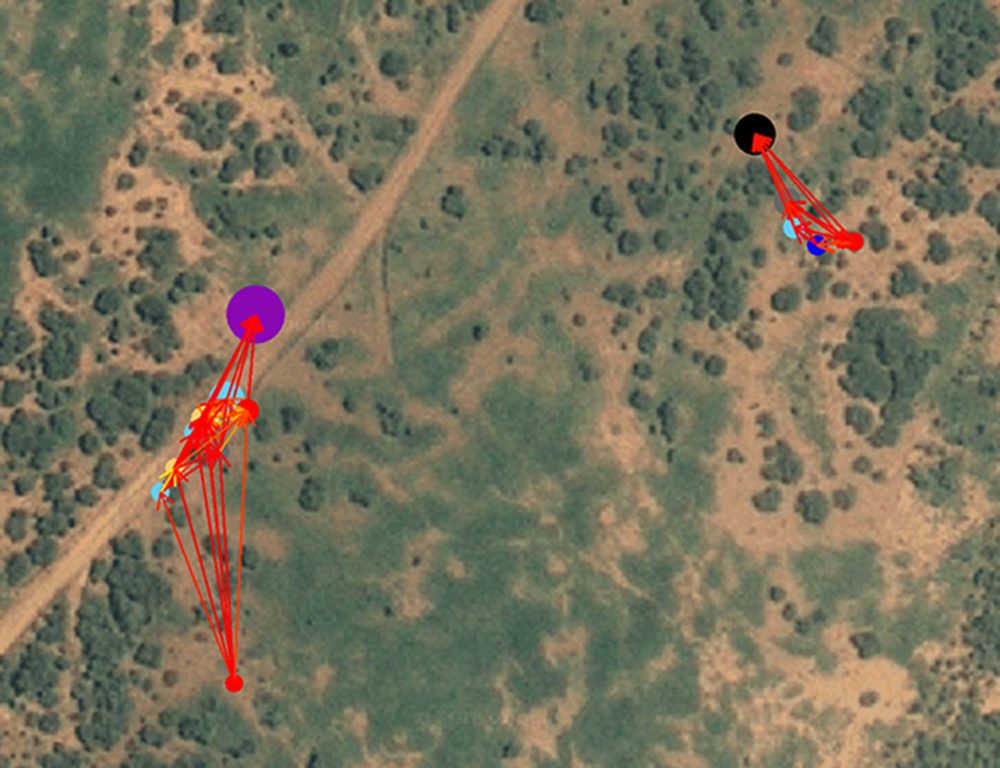}}
\subfigure[$t=350$]{\label{fig:baboon-exb}\includegraphics[width=0.32\columnwidth]{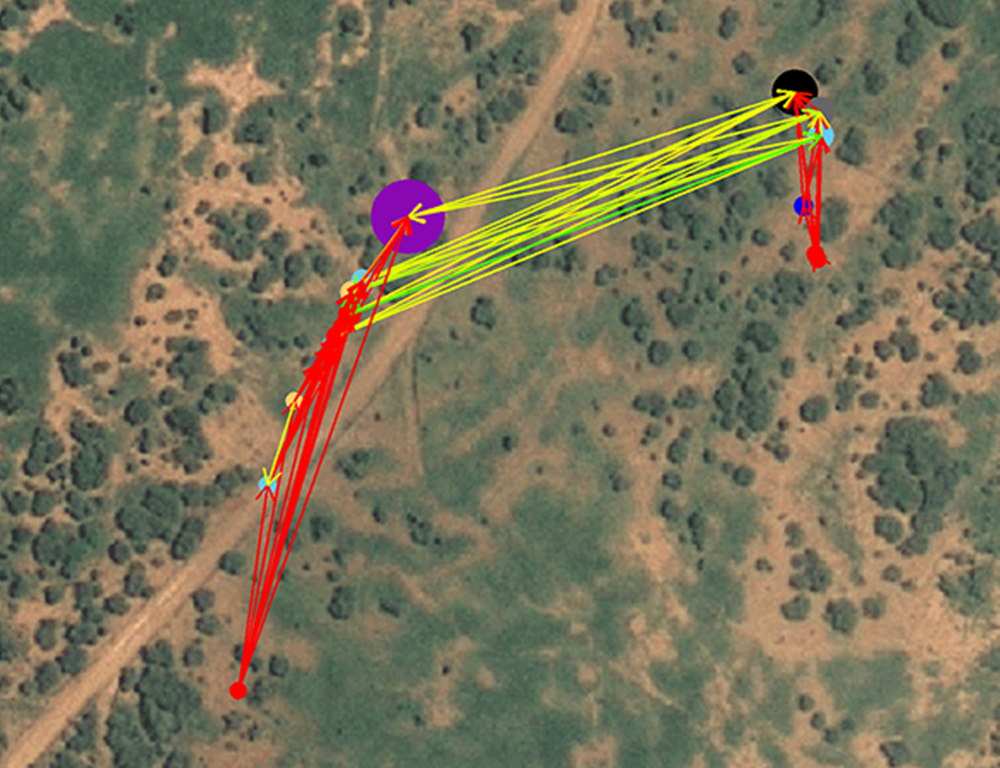}}
\subfigure[$t=400$]{\label{fig:baboon-exc}\includegraphics[width=0.32\columnwidth]{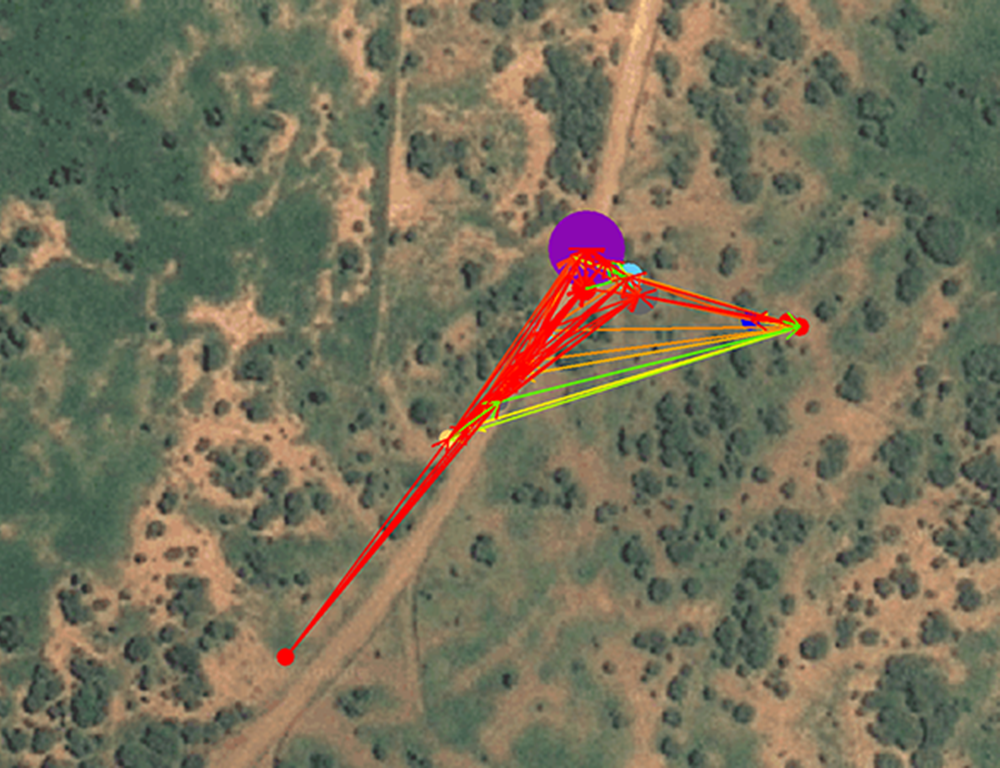}}
\caption{ The merging coordination event. (Top) Faction size ratios (Eq.~\ref{eq:FZratio}) of $\text{ID}(3)$ and $\text{ID}(18)$ factions. (Bottom) The GPS locations of individuals in the map over three different time steps ($t=300,350,400$), with the `following' network, and PageRank indicated by node size. $\text{ID}(3)$ is black and $\text{ID}(18)$ is purple. The red edges have higher edge weights than the light edges.}
\label{fig:BaboonMerging}
\end{figure}

We used a baboon dataset (see Section~\ref{subsubsec:baboon}) to demonstrate an example application of our framework to find transitions of coordinated events in real datasets. We focused on the dataset during the period when the merging of two groups happens on Aug 3, 2012, 08:49:01 AM. The length of the trajectories is 500 seconds. Figure~\ref{fig:BaboonMerging} illustrates the result when the merging happens. Before time $t=300$, a faction lead by $\text{ID}(3)$ (black node) starts moving in the same direction as the faction lead by $\text{ID}(18)$ (purple node). The process is measured by the Faction size ratios (Eq.~\ref{eq:FZratio}) of both factions, which increase over time. After $t=300$, $\text{ID}(3)$ faction is merging with $\text{ID}(18)$'s faction to become a single faction at $t=400$. After merging, because the faction of $\text{ID}(18)$ gains more members, its Faction size ratio (Eq.~\ref{eq:FZratio}) increases. Hence, by observing Faction size ratios lead by each individual, we can find merging events (or spiting events).

\section{Discussion}
In this paper, we formalized the \flip and provided an end-to-end general, unsupervised framework as the novel solution that can be used to study a wide range of coordinated activities.
The framework is competitive against a non-trivial baseline method in both simulated and real-world datasets. Moreover, we demonstrated that the framework can be used to identify merging events as well as factions and initiators at each time step in biological datasets. This example implies that our framework opens opportunities for scientists to ask questions about coordinated activities and is able to create scientific hypotheses and test them. Our framework is powerful and almost parameter free (we need only a similarity threshold $\sigma$ and time shift $\delta$ parameter). The scalability bottleneck is the DTW method used to compare time series. The existing DTW lower/upper bound techniques cannot be applied directly in our case since they only compute the distance between time series and not the actual wrapping path needed in our framework. With simpler and faster similarity computation, our framework can become highly computationally scalable. In the future, such more scalable approaches should be investigated. Another future work we plan to explore is the causality inference, which is closely related to leadership inference in the sense that initiators cause their followers' actions. We are  planning to report the Granger causality results for leadership inference in our next paper.  The code, datasets, and supplementary files that we used in this paper can be found at \cite{ShareSourcecode}. The new mFLICA code is in the form of R package~\cite{amornbunchornvej2020mflica}.

\balance
\bibliographystyle{abbrv}

\balance

\newpage
\section*{Supplementary}
\subsection*{Coordination measure}
Given a set of time series $\mathcal{U}$, a set of clusters $\mathcal{C}=\{H_1,\dots,H_n\}$ such that $\bigcup_k H_k=\mathcal{U}$,  we define a cluster membership indicator $\delta_{i,j} =1$ if time series $U_i$ and $U_j$ belong to the similar cluster, otherwise it is zero.  {\bf The average coordination measure} $\Psi$ of a set of clusters $\mathcal{C}$ is defined as follows:

\begin{equation*}
\Psi(\mathcal{C}) = \frac{\displaystyle\sum_{ U_i,U_j \in \mathcal{U},U_i \neq U_j }\mathrm{sim}_{max}(U_i,U_j)\delta_{i,j}}{\displaystyle\sum_{ U_i,U_j \in \mathcal{U},U_i \neq U_j }\delta_{i,j}}.
\end{equation*}

Note that $\Psi \in [0,1]$. If $\Psi$ is close to one, then all the time series within the same cluster are highly similar, with some time delay. This implies there exists a high degree of coordination within each clusters in this case. On the contrary, $\Psi \approx 0$ implies no coordination, on average.

\begin{theorem}
Given a set of time series $\mathcal{U}$ containing a set of $\sigma$-faction $\mathcal{F}=\{F_1,\dots,F_n\}$ where $\bigcup_{F_i \in \mathcal{F}}F_i=\mathcal{U}$, then, for all possible sets of clusters, $\mathcal{F}$ maximizes the average coordination measure $\Psi$. 
\end{theorem}
\begin{proof}
Reminding that for all pairs $U_i, U_j$ within any similar faction $F$, $\mathrm{sim}_{max}(U_i,U_j) \geq \sigma$. Hence, $\Psi(\mathcal{F})\geq \sigma$.\\

Case 1: let $H,J \in \mathcal{F}$, if we modify $\mathcal{F}$ by exchanging any time series $U_H \in H$ with $U_J \in J$ and call it $\mathcal{C}$, then we have:
 
\begin{equation*}
\Psi(\mathcal{F}) - \Psi(\mathcal{C}) = \frac{ S + S' }{\sum_{ U_i,U_j \in \mathcal{U},U_i \neq U_j }\delta_{i,j}}.
\end{equation*}

\begin{eqnarray*}
S = \sum_{U_i\in H\setminus\{U_H\} }\Big(\mathrm{sim}_{max}(U_i,U_H) - \mathrm{sim}_{max}(U_i,U_J) \Big) \\
S' = \sum_{ U_i \in J\setminus\{U_J\} }\Big(\mathrm{sim}_{max}(U_i,U_J)-\mathrm{sim}_{max}(U_i,U_H)\Big) \\																																				 
\end{eqnarray*}

For any $U_i \in H$, $\mathrm{sim}_{max}(U_i,U_H) \geq \sigma$ since $U_H \in H$. In contrast, because $U_J \notin H$, then $\mathrm{sim}_{max}(U_i,U_J) < \sigma$, which implies $S>0$. $S'>0$ for a similar reason.  Therefore, $\Psi(\mathcal{F}) - \Psi(\mathcal{C}) >0$.

Case 2: if we create $\mathcal{C}$ from  $\mathcal{F}$ by spiting a cluster $H \in \mathcal{F}$ to be $H_1 \subset H$ and $H_2 = H\setminus{H_1}$, then we have:

\begin{equation*}
\Psi(\mathcal{F}) - \Psi(\mathcal{C}) = \frac{ \sum_{U_i \in H_1,U_j \in H_2}\mathrm{sim}_{max}(U_i,U_j) }{|H_1||H_2|} \geq \sigma.
\end{equation*}

Case 3: we create $\mathcal{C}$ from $\mathcal{F}$ by merging any cluster $H \in \mathcal{F}$ with any $J \in \mathcal{F}$ such that $H \neq J$ to be $H'$. So, let 

\begin{equation*}
\Psi(\mathcal{F}) = \frac{X_{\mathcal{F}}}{S_{\mathcal{F}}} \geq \sigma, 
\end{equation*}
then 
\begin{equation*}
\Psi(\mathcal{C}) = \frac{X_{\mathcal{F}}+ \sum_{U_i \in H,U_j \in J}\mathrm{sim}_{max}(U_i,U_j) }{S_{\mathcal{F}}+|H||J|}.
\end{equation*}

By merging $H$ and $J$, we introduce pairs of time series across $H$ and $J$ to Equation~\ref{eq:coor-measure} such that $\mathrm{sim}_{max}(U_i,U_j) < \sigma$ since these pairs are not belong to the same faction. These pairs decrease the average of $X_{\mathcal{F}}$, which implies $\Psi(\mathcal{F}) > \Psi(\mathcal{C})$.

Since we shown that no matter how we edit $\mathcal{F}$, the average coordination measure $\Psi$ cannot increase, therefore, $\mathcal{F}$ maximizes the average coordination measure. 
\end{proof}

\subsection*{Time complexity}
Let $n$ be a number of time series, $\omega$ be a time window, $\delta$ be a shifting factor (we use $\delta=0.1\omega$), and $t^*$ be a total length of time series. By deploying DTW Sakoe Chiba band technique~\cite{sakoe1978dynamic} setting $\delta$ as a band limitation, the time complexity of computing a following network is $\mathcal{O}(n^2 \times \omega \times \delta)$. Since we need warping paths, not a distance, the upper/lower bounds tricks which are used to speed up DTW found in the time series literature cannot be applied here. The number of following networks we need to compute is $\frac{t^*}{\delta}$. In total, the time complexity of our framework is $\mathcal{O}(n^2 \times \omega \times t^*)$. Additionally, we might explore $k$ candidates of $\omega$ in order to find the optimal $\omega$. Since $k$ is a constant, the asymptotic time complexity of our framework also remains the same. This expensive cost is unavoidable and it makes our framework hard to be a scalable framework.  

\newpage
\subsection*{Comparison method}
From the best of our knowledge, there is no existing methods dealing with the \flip. The closest method that we can compare against is the flock model~\cite{will2016flock,andersson2008reporting}. We compared our framework against Volatility Collective Behaviors Model~\cite{will2016flock}, which has an assumption that all members in a similar group move toward the similar direction on a non-linear trajectory.  Hence, we modified the FLOCK framework to make it work in our setting as a baseline of comparison. In stead of using DTW to build following networks, we created FLOCK following networks. According to the work in \cite{andersson2008reporting}, the time series $A$ follows the time series $B$ at any time step $t$ if the angle of their direction vector from time $t-1$ to $t$ is less than the threshold $\beta$ as well as $B$ is in the front of $A$ with respect to $B$'s direction, as well as $A$ and $B$ must have their distance less than the threshold $\gamma$. The FLOCK following networks are built for all time steps. The rest of FLOCK framework is similar to our framework. We set the FLOCK parameters such that it can perform the best.

\subsection*{Centrality measures in multi-faction datasets}
In this section, we explore the use of centrality measures to infer faction initiators. We used 200 simulated datasets from the dictatorship model to conduct the analysis. For each dataset, we created a global static following network and used centrality measures on this network. In each dataset, we have 30 individuals and four of them are initiators. 

\begin{table}[th]
\centering
\caption{Jaccard similarity between top-4 ranking individuals from centrality measures and the ground truth set of four initiators in dictatorship model from 200 datasets.}
\label{table:JCsim}
\begin{tabular}{c|c|c|c|}
\cline{2-4}
                                        & \multicolumn{3}{c|}{Centrality methods} \\ \hline
\multicolumn{1}{|c|}{Event types}       & PageRank    & IN-Degree   & Closeness   \\ \hline
\multicolumn{1}{|c|}{Linear}            & 0.85        & 0.84        & 0.54        \\ \hline
\multicolumn{1}{|c|}{Merge/split} & 0.64        & 0.67        & 0.53        \\ \hline
\end{tabular}
\end{table}

The Jaccard similarity result between the top-4 ranking individuals from the centrality measures and the ground truth set of four initiators is in Table~\ref{table:JCsim}. PageRank and IN-Degree centrality perform well in dataset containing the simple linear coordination events while closeness centrality performs the worst. This is because initiators in this setting are supposed to have a higher number of followers than non-initiator individuals, which implies the higher ranking w.r.t. PageRank and In-Degree centrality. In contrast, initiators are not necessary close to their followers in the network, which made closeness centrality perform poorly. For the datasets that contain merge/split-coordination events, since there is a complicated dynamics of interactions among the factions, the simple centrality measures fail to capture the true initiators altogether.

\begin{table}[th]
\centering
\caption{Support of four initiators being in the list of top-4 ranking individuals from centrality measures in 100 datasets containing linear coordination events.}
\label{table:supLinear}
\begin{tabular}{c|c|c|c|}
\cline{2-4}
                                     & \multicolumn{3}{c|}{Centrality methods} \\ \hline
\multicolumn{1}{|c|}{Initiator's ID} & PageRank    & IN-Degree   & Closeness   \\ \hline
\multicolumn{1}{|c|}{ID1}            & 1           & 1           & 0.92        \\ \hline
\multicolumn{1}{|c|}{ID2}            & 1           & 1           & 0.66        \\ \hline
\multicolumn{1}{|c|}{ID3}            & 1           & 1           & 0.36        \\ \hline
\multicolumn{1}{|c|}{ID4}            & 0.39        & 0.37        & 0.20        \\ \hline
\end{tabular}
\end{table}

Table~\ref{table:supLinear} illustrates the result of supports of four initiators being in the list of individuals ranked top-4 by the centrality measures in linear-coordination datasets. Similarly, PageRank and In-Degree centrality perform well, while closeness centrality perform poorly. For the merge/split-coordination datasets,  Table~\ref{table:MS} shows that all centrality measures perform poorly to infer ID2 and ID4 initiators while they perform well to include ID1 and ID3 in their top-4 ranking lists. This is because ID1 and ID3 spent significantly more time leading their factions than ID2 and ID4.

\begin{table}[bh]
\centering
\caption{Support of four initiators being in the list of top-4 ranking individuals from centrality measures in 100 datasets containing merging/splitting coordination events.}
\label{table:MS}
\begin{tabular}{c|c|c|c|}
\cline{2-4}
                                     & \multicolumn{3}{c|}{Centrality methods} \\ \hline
\multicolumn{1}{|c|}{Initiator's ID} & PageRank    & IN-Degree   & Closeness   \\ \hline
\multicolumn{1}{|c|}{ID1}            & 1           & 1           & 1           \\ \hline
\multicolumn{1}{|c|}{ID2}            & 0.29        & 0.47        & 0.20        \\ \hline
\multicolumn{1}{|c|}{ID3}            & 1           & 1           & 0.83        \\ \hline
\multicolumn{1}{|c|}{ID4}            & 0.28        & 0.19        & 0.08        \\ \hline
\end{tabular}
\end{table}

In conclusion, these results emphasize the need of  a dynamic following network approach to deal with the complicated problem  of inferring the initiator of a faction.  

\newpage
\subsection*{The  pseudo  codes}
\label{sec:pseudo codes}

\IncMargin{1em}
\begin{algorithm2e}
\label{algo:CreateFollowingNetwork}
\caption{CreateFollowingNetwork}
\SetKwInOut{Input}{input}\SetKwInOut{Output}{output}
\Input{A time series set $\mathcal{Q}=\{Q_1,\dots,Q_n \}$ and a threshold $\sigma$}
\Output{A $n\times n$ adjacency matrix $E$ }
\BlankLine
$E_{i,j} = 0, \forall i,j \in \{1,\dots,n\}$\;
\For{$i\leftarrow 1$ \KwTo $n$}{
\For{$j\leftarrow i+1$ \KwTo $n$}{\label{forins}
$U\leftarrow Q_i$ and $W \leftarrow Q_j$\; 
$P_{U,W}\leftarrow DTW(U,W)$ \;  
\uIf{$\mathrm{s}(P_{U,W}) \geq \sigma $}{
$E_{j,i}=|\mathrm{s}(P_{U,W})|$\; 
}
\uElseIf{$\mathrm{s}(P_{U,W}) \leq -\sigma $}{
$E_{i,j}=|\mathrm{s}(P_{U,W})|$\; 
}
\Else{
$E_{i,j} = 0$\; 
}

}
}
\end{algorithm2e}\DecMargin{1em}

\IncMargin{1em}
\begin{algorithm2e}
\caption{CreateDyFollowingNetwork}
\SetKwInOut{Input}{input}\SetKwInOut{Output}{output}
\Input{A time series set $\mathcal{U}$, $\omega$, $\delta$, and $\sigma$}
\Output{A $n\times n \times t^*$ adjacency matrix $E^*$.}
\BlankLine
$K \leftarrow (t^* - \omega)/ \delta$ \;
\For{$i\leftarrow 1$ \KwTo $K$}{

\textcolor{cyan}{\tcc*[h]{current time interval} }
$w(i)=[(i-1)\times \delta,(i-1)\times \delta + \omega]$ \;
 \textcolor{cyan}{\tcc*[h]{SubTimeSeries($U,w(i)$) returns all sub time series in $U$ within the interval $w(i)$}}
$Q \leftarrow$SubTimeSeries($U,w(i)$)\;
$E\leftarrow$CreateFollowingNetwork($Q,\sigma$) \;

 \textcolor{cyan}{\tcc*[h]{Set all edges within the time interval $[(i-1)\times\delta,i\times\delta]$ to be similar}}

$E^*_{ t \in [(i-1)\times\delta,i\times\delta]}\leftarrow E$ \;
}
$Q \leftarrow$SubTimeSeries($U,[K\times\delta,t^*]$)\;
$E\leftarrow$CreateFollowingNetwork($Q,\sigma$) \;
$E^*_{ t \in [K\times\delta,t^*]}\leftarrow E$ \;
\label{algo:CreateDyFollowingNetwork3}
\end{algorithm2e}\DecMargin{1em}

\IncMargin{1em}
\begin{algorithm2e}
\caption{FindFactionsAndInitiators}
\SetKwInOut{Input}{input}\SetKwInOut{Output}{output}
\Input{An adjacency matrix $E^*$ of dynamic network}
\Output{ A time series of faction sets $\mathcal{F}^*$, and a time series of initiator sets $\mathcal{L}^*$ }
\BlankLine
\For{$i\leftarrow 1$ \KwTo $t^*$}{
\textcolor{cyan}{\tcc*[h]{Get a matrix at time $t=i$} }

$E \leftarrow E^*_{ t = i}$ \; 
\textcolor{cyan}{\tcc*[h]{FindInitiators($E $) returns all nodes which have zero outgoing degree}}

$\mathcal{L} \leftarrow$FindInitiators($E $) \;
$\mathcal{F} = \emptyset$ \;
\For{$l \in \mathcal{L}$}{
\textcolor{cyan}{\tcc*[h]{FindReachNodeFrom($E,l$) returns all nodes which have any directed path to $l$}}

$F_l\leftarrow$FindReachNodeFrom($E,l$) \;
$\mathcal{F} = \mathcal{F} \cup \{F_l\}$
}
$\mathcal{F}^*_{t=i} = \mathcal{F}  $ and $\mathcal{L}^*_{t=i} = \mathcal{L}  $
}
\label{algo:FindFactionsAndInitiators}
\end{algorithm2e}\DecMargin{1em}

\end{document}